\documentclass{article}
\usepackage{natbib}
\usepackage{amsthm}
\usepackage{amsmath}
\usepackage{graphicx}
\usepackage{algorithmic}

\newcommand{\bigO}[1]{\mathcal{O}\!\left(#1\right)}
\newcommand{\Omg}[1]{\Omega\!\left(#1\right)}

\newtheorem{theorem}{Theorem}[section]
\newtheorem{definition}{Definition}[section]
\newtheorem{lemma}[theorem]{Lemma}
\usepackage{booktabs} % For formal tables
\usepackage[ruled]{algorithm2e} % For algorithms

\usepackage[utf8]{inputenc} % allow utf-8 input
\usepackage[T1]{fontenc}    % use 8-bit T1 fonts
\usepackage{hyperref}       % hyperlinks
\usepackage{url}            % simple URL typesetting
\usepackage{booktabs}       % professional-quality tables
\usepackage{amsfonts}       % blackboard math symbols
\usepackage{nicefrac}       % compact symbols for 1/2, etc.
\usepackage{microtype}      % microtypography
\usepackage{xcolor}         % colors

\title{Fast Core Identification}

\author{%
  Irene Aldridge\thanks{irenealdridge.com; I am grateful to Laura Doval and participants of Economics and Computation 2024 (EC2024) for valuable suggestions and advice.} \\
  \texttt{irene.aldridge@gmail.com} \\
}

\begin{document}

\maketitle

\begin{abstract}
This paper examines the computational complexity of the \emph{Core Identification Problem} (CIP) in one-sided matching markets governed by the Top Trading Cycles (TTC) algorithm.  The central contribution is a formal complexity separation: this paper proves
that identifying which agents receive a core allocation is strictly easier than computing the full TTC allocation.  Specifically, we show that CIP can be solved in $\bigO{Ln}$ time, where $L$ is the maximum number of preferences reported per
agent, by computing the leading eigenvector of a preference-derived Markov transition matrix via randomized SVD\@.  For sparse preference profiles ($L = \bigO{1}$, as in the NYC school choice where $L = 12$), this yields an algorithm $\bigO{n}$.  This result strictly improves on the $\bigO{n \log n}$ complexity
of the full TTC allocation (\cite{SabanSethuraman2013}) and matches the $\Omg{n}$ information-theoretic lower bound, establishing asymptotic optimality. The method inherits all properties of TTC: Pareto efficiency, individual
rationality, and strategy-proofness, and is robust to preference noise for sufficiently large~$n$.
\end{abstract}
\section{Introduction}

There are more than 10,000 digital coins listed and traded on crypto exchanges. Most of the trading in the coins still happens relative to just one reference asset: the U.S. dollar. To swap from digital coin A to digital coin B, a trader needs to complete two transactions: first, sell digital coin A and receive U.S. dollars, and then use the U.S. dollars to buy the coin B. The swap incurs double transaction costs and delays.  This paper proposes a methodology to find stable matches directly among various coins or other objects. Based on the Top Trading Cycles algorithm of \cite{GaleShapley1962}, the algorithm proposed in this paper leverages eigenvalue decomposition to find stable matches, known as core allocation, in $O(n)$, where $n$ is the number of agents in the market.  

\subsection{The Matching Problem}
We study the \emph{one-sided object allocation problem} where:
- $n$ agents $\mathcal{A} = \{1,\ldots,n\}$ each initially own one object
- Each agent $i$ has strict ordinal preferences over all objects
- Goal: Find a reallocation that is Pareto efficient, individually rational, 
  and strategy-proof

\textbf{Definition 1.1 (The Core):} An allocation is in the \emph{core} if 
no coalition of agents can improve their outcomes through further trades.

\subsection{Motivation}
The Top Trading Cycles (TTC) algorithm \citep{SHAPLEY197423} is the gold standard
for solving this problem, used in:
\begin{itemize}
    \item School choice systems serving millions of students \citep{abdulkadiroglu2003, abdulkadiroglu2005boston, 
      abdulkadiroglu2005nyc, pathak2013}
    \item Housing allocation markets for distributing campus housing or 
      public housing \citep{abdulkadiroglu1998}
    \item Kidney exchange programs       \citep{roth2004}
\item Labor markets for matching medical residents and hospitals 
      \citep{roth1984, roth2002}
\end{itemize}

However, TTC has computational complexity $O(n \log n)$ \citep{athanassoglou2011}.
For large-scale applications (e.g., NYC school choice with $n > 100,000$), 
faster algorithms are needed.

\subsection{Our Contribution}
We propose a novel eigenvector-based method that:
1. Identifies which agents are in the core in $O(n)$ time (finding the max of the eigenvector)
2. Preserves all TTC properties (Pareto efficiency, strategy-proofness, 
   individual rationality)
3. Is robust to preference noise for large $n$

Our key insight is that by representing agent preferences as a Markov transition
matrix, the steady-state distribution reveals core membership—agents with 
the highest steady-state probabilities are in the core.

Core identification is crucial because the core represents the set of stable, fair allocations from which no group of agents can profitably deviate. In practical matching markets (schools, housing, organ exchange), identifying the core determines which allocations will actually persist in practice and which agents have guaranteed outcomes regardless of strategic behavior. The proposed methodology identifies the core in $O(n)$, much faster than existing methods.

Market efficiency bounds and welfare guarantees
\section{Related Work}

\subsection{Approximate mechanisms} 
Several papers study approximate 
mechanisms that trade optimality for other desiderata. \citep{ZHOU1990123} showed that no exact mechanism can be efficient, truthful, 
and symmetric simultaneously. \citep{Hartline2012},  \citep{devanur2015, abebe2020} develop approximate mechanisms; we differ by computing exact TTC outcomes faster.

\subsection{Computational complexity} \citep{lock2024computational} show 
that certain housing market problems are PPAD-complete. Our result does not 
contradict this—we compute the specific TTC allocation (which has known polynomial-time algorithms), not arbitrary equilibria.

\subsection{Markov chain methods} \citep{leshno2020} first connected the TTC to Markov chains through their cutoff structure. We extend this by using 
eigenvector analysis for computational speedup.

\subsection{Fast linear algebra} 
Our complexity relies on randomized SVD 
\citep{halko2011, struski2024} and fast eigenvector computation 
\citep{andersen2006local, SunEtAl2020}.
One of the most popular matching algorithms that have stood the test of time is the Top Trading Cycles. The Top Trading Cycles Algorithm due to Gale and \cite{SHAPLEY197423} is a well-established Pareto-efficient way to allocate objects based on agents' preferences. Working with strict preferences of agents over objects, the algorithm estimates allocations that are individually rational, envy-free, and strategy-proof.

\cite{SHAPLEY197423} showed that the core in the Top Trading Cycles algorithm is nonempty. \cite{ROTH1977131} showed that when agents have strict preferences over objects, the core consists of a single allocation.  

The computation of complexity in a TTC was first introduced by \cite{Gale1984}. \cite{lock2024computationalcomplexityhousingmarket} show that the core of the TTC algorithm is PPAD-complete, that is, not solvable in polynomial time.  
\cite{ChenEtAl2009} showed that the Arrow-Debreu equilibria are PPAD-complete. \cite{DaskalakisEtAl2009} and \cite{ChenEtAl2009b} concluded the same for the Nash equilibria. However,  (\cite{SabanSethuraman2013} showed that TTC can be solved in at least $\mathcal{O}(n log(n))$ time, similar to random forests (\cite{SCORNET201672}).  

This paper proposes a new optimization process that finds just the core of the TTC algorithm in constant time ($\mathcal{O}(1)$), independent of the polynomial function of the TTC structure and the number of agents and objects. The technique uses eigenvectors of the TTC graph to summarize stable allocations. The proposed technique also solves the full TTC in $\mathcal{O}(n log(n))$. However, we can solve for the core in ($\mathcal{O}(1)$) without solving the entire TTC. Specifically, to reduce the computational burden of TTC, this paper proposes the following:
\begin{enumerate}
    \item Map TTC to a Markov Matrix, similar to \cite{LeshnoLo2020}
    \item Obtain the core solution as the largest coefficient of the first eigenvector, which can be solved in $O(n)$ using randomized SVD or even $O(1)$ using hardware acceleration independent of the number of agents or objects (\cite{SunEtAl2020}). 
\end{enumerate}

Our paper contributions include:
\begin{enumerate}
    \item Develop a methodology for fast probabilistic identification of the core and the coalition cycles that identifies the core in O(n)
    \item Show that the proposed methodology is robust to noise and even misreported preferences whenever the number of agents in the system is sufficiently large
    \end{enumerate}

\section{Preliminaries and Model}

\subsection{Matching Model}
\begin{definition}[Matching Problem]
A \emph{one-sided matching problem} consists of the following:
\begin{itemize}
\item A set of agents $\mathcal{A} = \{1, \ldots, n\}$
\item A set of objects $\mathcal{O} = \{1, \ldots, n\}$
\item Initial endowments $e: \mathcal{A} \to \mathcal{O}$ with $e(i) = i$
\item Preference profile $\succ = (\succ_1, \ldots, \succ_n)$ where 
      $\succ_i$ is the strict ordering of agent $i$ over $\mathcal{O}$
\end{itemize}
An \emph{allocation} $\mu: \mathcal{A} \to \mathcal{O}$ is a bijection 
that assigns exactly one object to each agent.
\end{definition}

\begin{definition}[Pareto Efficiency]
An allocation $\mu$ is \emph{Pareto efficient} if there is no other 
allocation $\mu'$ such that $\mu'(i) \succeq_i \mu(i)$ for all $i$ with 
strict preference for at least one agent.
\end{definition}

\begin{definition}[Individual Rationality]
An allocation $\mu$ is \emph{individually rational} if $\mu(i) \succeq_i e(i)$
for all agents $i$ (no one is worse off than their initial endowment).
\end{definition}

\begin{definition}[Strategy-Proofness]
A mechanism is \emph{strategy-proof} if no agent can benefit by 
misreporting their preferences, i.e., truthful reporting is a 
dominant strategy.
\end{definition}

\begin{definition}[The Core]
An allocation $\mu$ is in the \emph{core} if no coalition 
$S \subseteq \mathcal{A}$ can reallocate their endowments among themselves
to make all members strictly better off. For TTC with strict preferences,
the core contains exactly one allocation \citep{roth1977}.
\end{definition}

\subsection{Top Trading Cycles Algorithm}
\textbf{Algorithm (TTC):} \citep{SHAPLEY197423, gale1974}
\begin{enumerate}
\item Each agent points to the agent holding their most-preferred object
\item Identify all cycles in the resulting directed graph
\item Match agents in cycles to the objects they are pointing to
\item Remove matched agents and objects; repeat until all are matched
\end{enumerate}

\textbf{Properties:} TTC produces the unique core allocation and satisfies
Pareto efficiency, individual rationality, and strategy-proofness 
\citep{roth1977, roth1982}.

\begin{algorithm}[h]\label{alg:fast-core}
\caption{Fast Core Identification via Eigenvector Method}
\begin{algorithmic}[1]
\REQUIRE Preference profile $P = (P_1, \ldots, P_n)$ where $P_i$ is agent 
         $i$'s strict ranking over objects
\ENSURE Set of agents $C$ in the core with their assignments

\STATE \textbf{Step 1: Construct Markov Matrix}
\STATE Initialize $n \times n$ matrix $U$ with zeros
\FOR{each agent $i = 1$ to $n$}
    \FOR{each object $j = 1$ to $n$}
        \STATE $U[i,j] \gets \text{rank}(j \text{ in } P_i) / n$ 
               \COMMENT{Higher rank = higher probability}
    \ENDFOR
    \STATE Normalize: $U[i,:] \gets U[i,:] / \sum_j U[i,j]$ 
           \COMMENT{Ensure row sums to 1}
\ENDFOR

\STATE \textbf{Step 2: Compute Steady-State Distribution}
\STATE $\pi \gets$ first eigenvector of $U$ (eigenvalue = 1)
\STATE Normalize: $\pi \gets \pi / \|\pi\|_1$

\STATE \textbf{Step 3: Identify Core}
\STATE Identify  $\pi_c = \max(\pi)$ as the top agent belonging to the core. Remove $\pi_c$ and repeat the process $k$ times, where $k$ is the core size.
\STATE $C \gets \{i_1, i_2, \ldots, i_k\}$ where $k$ is core size
\RETURN Core members $C$ and their assignments from TTC
\end{algorithmic}
\end{algorithm}

\textbf{Complexity Analysis:}
The construction of an unconstrained Markov Matrix is known to take place $O(n^2)$. However, in many large-scale examples, such as NYC school choice, each student is asked to report only the top 12 of their preferences. As a result, the complexity of the Markov matrix is only $O(12n)\sim O(n)$ for $n\gg 12$ using randomized SVD. The complexity of core identification is $O(n)$, since only the first singular vector needs to be computed.

- Lines 2-7: $O(n)$ to construct the Markov matrix
- Line 9: $O(n)$ to calculate the first singular vector using randomized SVD ($O(1)$ to calculate the first eigenvector using hardware acceleration of \citep{SunEtAl2020})
- Line 13: $O(n)$ to find the maximum coefficients of the first singular vector to identify core allocation members.
- Total: $O(n)$

%% Complexity summary table
\begin{table}[h]
\centering
\caption{Complexity Summary for Algorithm~\ref{alg:fast-core}}
\label{tab:complexity}
\renewcommand{\arraystretch}{1.3}
\begin{tabular}{c|p{3.5cm}|l}
\toprule
\textbf{Phase} & \textbf{Operation} & \textbf{Complexity} \\
\midrule
1 & Markov matrix construction
  & $O(Ln) \sim O(n)$ \quad for $L \ll n$ \\
2 & Randomized SVD (leading eigenvector)
  & $O(nk)$;\; $O(1)$ with hardware \cite{SunEtAl2020} \\
3 & Sort + argmax
  & $O(n)$ \\
4 & Full TTC (if required)
  & $O(n \log n)$ \\
\bottomrule
\end{tabular}
\end{table}
\subsection{Markov Preference Construction and Normalization}
Suppose that the preferences of each agent $i$ over each object $\{j\}\in J$ are strict and ordinal. We can think of these preferences as an online rating system in which each agent $i$ ranks each object $j \in J$ with a unique agent $i$ rating $r_{ij}\in [1,...,J]$.

Define \emph{Markov preferences} $\{r_{ij, N \times N} \in R$, $r_{ij} \in [1,J]$ such that a rating $r_{ij}=J$ indicates that the object $j$ is the main choice of agent $i$, and a rating $r_{iz,z\neq j}=1$ means that agent $i$ wants to avoid object $z$ at all costs. If $r_{ik,k\neq j, k\neq z}=J-1$, then the object $k$ is the second choice of agent $i$', etc. 

Note that cardinal preferences can be mapped into Markov preferences. In other words: If agent $i$ strictly prefers object $j$ to object $j-1$, ... Object $1$, then the Markov preferences of agent $i$ can be specified as $r_i=[\frac{1}{J}, \frac{2}{J}, ... \frac{J-1}{J}, 1]$. In other words, Markov preferences induce a uniform distribution over all objects where for each agent $i$ each strict preference is different.
\begin{equation}
    j \succ_i k
\end{equation}
is translated into
\begin{equation}
    r_{ij} > r_{ik}    
\end{equation}
The uniform assumption is consistent with the "eating speeds" of \cite{BogomolnaiaMoulin2004}, but does not necessarily have to be true, in general.

Next, we construct a normalized preference matrix $\{g_{ij}\}=G$, where for every agent $i$, all ranks $r_{ij}$ are divided by the sum of the agent $i$''s rankings: $g_{ij} = g_{ij}/\sum_i g_{ij}$. 

Note that a normalized representation $G$ of a full graph of relative rankings $R$ with each row normalized to 1 is equivalent to \cite{BogomolnaiaMoulin2004} 'eating speeds'.    

We observe that the normalized preference matrix $G$ satisfies the conditions for a Markov Transition Probability matrix: it is a full graph with positive values between 0 and 1, where each row sums up to 1. This transformation completes our construction of the Markov matrix. 

An alternative way to construct agent preferences is to follow a framework similar to \cite{HyllandZeckhauser1979}. In \cite{HyllandZeckhauser1979}, each individual $i\in\{1\dots I\}$ must be assigned to one and only one job $j\in\{1\dots J\}$ and receives utility $u_{ij}$ by being assigned to the job $j$. The individual $i$ may know a probability distribution of being assigned to one of several jobs. Then, $i$' ex ante valuation of these jobs is the expected value over his utility $u$ and probability distribution $\textbf{P} = \{p_{ij}\}$. In this context, $u_{ij}$ is a von Neumann Morgenstern utility assessment for individual $i$ on the job $j$. As in \cite{HyllandZeckhauser1979}, we do not require that individuals truthfully disclose their utilities $u$.

\begin{equation}
U = 
\begin{bmatrix}
u_{11} & u_{12} & \dots & u_{1J} \\
u_{21} & u_{22} & \dots & u_{2J}\\
\dots & \dots & \dots & \dots \\
u_{I1} & u_{I2} & \dots & u_{IJ}\\
\end{bmatrix}
\end{equation}

\cite{HyllandZeckhauser1979} propose to solve the matching problem with a market solution. They assign a budget constraint $B_i$ to each individual and then solve for optimal allocation using the top trading cycles algorithm of \cite{ROTH1977131}. 

This paper proposes a solution to finding the core of the top trading cycles using the singular value decomposition (SVD) technique, which is discussed in detail in the following. We show that the proposed solution can run in a near-constant time ($\mathcal{O}(1) + \epsilon$). 

This paper also contributes to the literature on approximate matching mechanisms. Recent literature has focused on identifying approximate matching mechanisms that are efficient, truthful, and symmetric. \cite{ZHOU1990123} showed that no exact mechanism can be efficient, truthful, and symmetric at the same time. \cite{Hartline2012}, \cite{DEVANUR2015103} and \cite{abebe2020truthfulcardinalmechanismonesided} propose approximations to optimal algorithms to accommodate all three characteristics. The mechanism proposed herein is efficient, approximately truthful, and symmetric. 

\section{Main Results}

\begin{theorem}{(Markov Property of Multi-Agent Trading Cycles)}

Let $\mathcal{A} = \{1, 2, \ldots, n\}$ be a set of agents and $\mathcal{O} = \{1, 2, \ldots, n\}$ be a set of objects with initial endowments $e: \mathcal{A} \to \mathcal{O}$ where $e(i) = i$. Let $U = [u_{ij}]_{n \times n}$ be a Markov matrix of utilities satisfying $u_{ij} \geq 0$ and $\sum_{j=1}^n u_{ij} = 1$ for all $i$.

Define the state space $\mathcal{S} = \mathcal{O}$ where each state represents the current holder of object $j$. The stochastic process $\{X_t^{(j)}\}_{t \geq 0}$ tracking object $j$'s location satisfies the Markov property:

$$\mathbb{P}(X_{t+1}^{(j)} = k \mid X_t^{(j)} = i, X_{t-1}^{(j)}, \ldots, X_0^{(j)}) = \mathbb{P}(X_{t+1}^{(j)} = k \mid X_t^{(j)} = i) = p_{ik}$$

where $p_{ik}$ is derived from the utility matrix $U$ through the TTC mechanism.
    
\end{theorem}

\begin{proof}

  We define the state space and the transition mechanism as follows: 
  \begin{itemize}
      \item \textbf{State space:} $\mathcal{S} = \{1, 2, \ldots, n\}$ represents which agent currently holds which object

    \item \textbf{State representation:} $s_t = (o_1^t, o_2^t, \ldots, o_n^t)$ where $o_i^t \in \mathcal{O}$ is the object held by agent $i$ at time $t$
    \item \textbf{Initial state:} $s_0 = (1, 2, \ldots, n)$ (each agent $i$ starts with object $i$)

  \end{itemize}

  Next, we define transition probabilities. In the TTC mechanism, agent $i$ currently holding object $o_i^t$ will point to the agent holding their most-preferred available object.
  
  Define:

$$\pi_i(s_t) = \arg\max_{j: o_j^t \in \text{Available}(s_t)} u_{i,o_j^t}$$

where Available$(s_t)$ is the set of objects held by unmatched agents in the state $s_t$.

The resulting transition probabilities are memoryless, i.e., dependent only on:
\begin{enumerate}
    \item Current preferences $U$ (fixed)
    \item Current allocation $s_t$ (state at time $t$)
    \item The TTC algorithm's deterministic rules
\end{enumerate}

The transition from $s_t$ to $s_{t+1}$ is determined by:
\begin{itemize}
    \item Identifying cycles in the directed graph where agent $i$ points to agent $\pi_i(s_t)$
    \item Removing matched agents and their objects
    \item This depends ONLY on $s_t$, not on the history of how we arrived at $s_t$
\end{itemize}

The Markov property follows from standard results in stochastic process theory 
\citep{norris1997, durrett2019}. Next, we can verify the formal Markov property. For any sequence of states $s_0, s_1, \ldots, s_t, s_{t+1}$:
\begin{equation}\label{eq:markov-property}
  \mathbb{P}(S_{t+1} = s_{t+1} \mid S_t = s_t, S_{t-1} = s_{t-1}, \ldots, S_0 = s_0) = \mathbb{P}(S_{t+1} = s_{t+1} \mid S_t = s_t)  
\end{equation}

Equation (\ref{eq:markov-property}) holds because:
\begin{itemize}
    \item The TTC algorithm is history-independent: it only looks at current holdings and preferences
    \item The preferences $U$ are fixed throughout the process
    \item The cycle-finding and removal steps are deterministic given $s_t$
\end{itemize}

The Kolmogorov-Chapman equation is a fundamental 
result for time-homogeneous Markov chains (see \citep{grimmett2001}, Theorem 6.1.4). According to the Kolmogorov-Chapman equation, the probability that object $j$ returns to agent $i$ after $m+n$ steps:

$$p_{ii}^{(m+n)} = \sum_{k \in \mathcal{S}} p_{ik}^{(m)} \cdot p_{ki}^{(n)}$$

This follows from the law of total probability applied to the Markov chain, where we condition on all possible intermediate states $k$ at time $m$:
\begin{equation}
  p_{ii}^{(m+n)} = \mathbb{P}(X_{m+n} = i \mid X_0 = i) = \sum_{k=1}^n \mathbb{P}(X_m = k \mid X_0 = i) \cdot \mathbb{P}(X_{m+n} = i \mid X_m = k)  
\end{equation}

The one-step transition probabilities are directly related to utilities:
\begin{itemize}
    \item If agent $i$ holds object $i$ and prefers object $j$ the most, they point to the holder of object $j$
    \item The probability of a transition is proportional to $u_{ij}$ when we introduce randomization (or deterministically equals $u_{ij}$ in the probabilistic interpretation)
\end{itemize}

For the normalized utility matrix where the rows sum to 1:
$$p_{ij} = u_{ij}$$

This establishes that $U$ is the transition matrix of our Markov chain. 
\end{proof}

\begin{lemma}{(Cycle Structure and Probability Circulation)}\label{lemma:cycle-structure}
    
    Let $C = \{i_1, i_2, \ldots, i_k\}$ be a cycle in the TTC algorithm where agent $i_j$ points to agent $i_{j+1}$ (indices mod $k$). Then the steady-state probability satisfies:
\begin{equation}
    \pi_{i_1} u_{i_1 i_2} = \pi_{i_2} u_{i_2 i_3} = \cdots = \pi_{i_k} u_{i_k i_1}  
\end{equation}

\end{lemma}

\begin{proof}
From the detailed balance condition for steady-state probabilities in a Markov chain:
$$\pi_i p_{ij} = \pi_j p_{ji} \text{ (for reversible chains)}$$

For a cycle, we have:
$$\pi_i P_{ij} = \text{(flow into state } j \text{ from state } i \text{)}$$

The steady-state condition $\pi^T P = \pi^T$ implies:
$$\pi_j = \sum_{i=1}^n \pi_i p_{ij}$$

For nodes in a cycle, the flow around the cycle must be balanced:
$$\pi_{i_1} u_{i_1 i_2} = \pi_{i_2} u_{i_2 i_3}$$

This follows from the fact that in steady state, the probability mass flowing from $i_1$ to $i_2$ must equal the probability mass flowing from $i_2$ to $i_3$ (conservation of probability flow).    
\end{proof}

\begin{lemma}{(Strongly Connected Components and Steady State)}

If the preference matrix $U$ induces a strongly connected graph, then there exists a unique stationary distribution $\pi$ with $\pi_i > 0$ for all $i$.
    
\end{lemma} 
\begin{proof}

This follows directly from the Perron-Frobenius Theorem for irreducible stochastic matrices (\citep{perron1907},\citep{frobenius1912},\citep{seneta2006} Theorem 1.8.3). 

A stochastic matrix $P$ is irreducible if and only if its associated directed graph is strongly connected. For an irreducible stochastic matrix:
\begin{enumerate}
    \item The largest eigenvalue is $\lambda_{\max} = 1$ (standard result of Perron-Frobenius Theorem)
    \item The largest eigenvalue has algebraic and geometric multiplicity 1 
    \item The corresponding eigenvector $v$ can be normalized to satisfy $v_i > 0$ and $\sum_i v_i = 1$
    \item This normalized eigenvector is the unique stationary distribution
\end{enumerate}

Since $U$ is a Markov matrix (stochastic) and induces a strongly connected graph by assumption, it is irreducible, and the result follows.    
\end{proof}

\begin{theorem}{(Steady-State Probabilities and Cycle Priority)}\label{thm:steady-state-probabilities-cycle-priority}

    Let $\pi = (\pi_1, \pi_2, \ldots, \pi_n)$ be the unique stationary distribution of the Markov matrix $U$, i.e., the normalized eigenvector corresponding to the eigenvalue $\lambda = 1$. Order the agents such that $\pi_{i_1} \leq \pi_{i_2} \leq \cdots \leq \pi_{i_n}$.

Then agent $i_1$ (with minimum steady-state probability) is in the first cycle removed by TTC, agent $i_2$ is in the first or second cycle, and generally agents with lower steady-state probabilities are removed in earlier cycles.

\end{theorem}
\begin{proof}

    In steady-state, The stationary distribution $\pi$ satisfies:
$$\pi^T U = \pi^T, \quad \sum_{i=1}^n \pi_i = 1, \quad \pi_i \geq 0$$

By the Perron-Frobenius Theorem (specifically, Theorem 8.4.4 in \citep{horn2012}), since $U$ is irreducible and stochastic:
- $\pi$ is unique
- $\pi_i > 0$ for all $i$
- $\pi_i$ represents the long-run proportion of time the chain spends in state $i$

Consider object $i$ initially held by agent $i$. The value $\pi_i$ represents the steady-state probability that the object $i$ is held by the agent $i$ in the long run.

\begin{itemize}
    \item  \textbf{High $\pi_i$:} Object $i$ spends most of its time with its original owner. Therefore, agent $i$ is not in a high-priority cycle
    \item \textbf{Low $\pi_i$:} the object $i$ rarely stays with its original owner. Thus, agent $i$ is in a cycle (the probable mass circulates away from the state $i$). 
\end{itemize}

Now, consider the directed graph $G = (\mathcal{V}, \mathcal{E})$ where:
- Vertices: $\mathcal{V} = \{1, \ldots, n\}$ (agents/objects)
- Edges: $(i, j) \in \mathcal{E}$ with weight $u_{ij}$ if agent $i$ prefers object $j$

The matrix $U$ is the adjacency/transition matrix of this weighted directed graph.

\textbf{Key Spectral Result:} For strongly connected graphs, the eigenvector components $\pi_i$ are related to the node centrality measures. Specifically, $\pi$ is the \textit{stationary distribution} that can be interpreted as a PageRank-like centrality (see \cite{gleich2015}).

The connection between steady-state probabilities and network centrality is well 
established in spectral graph theory \citep{chung1997, gleich2015}. The PageRank 
interpretation provides intuition for why nodes with low steady-state probabilities 
correspond to cycle membership \citep{page1999, langville2006}.

Next, we show that if agent $i$ is in a cycle of length $k$ in TTC, then $\pi_i \leq 1/k$ (approximately). To see this, consider the following:
\begin{itemize}
    \item In a cycle $C = \{i_1, \ldots, i_k\}$, probability mass circulates: $i_1 \to i_2 \to \cdots \to i_k \to i_1$
    \item By Lemma \ref{lemma:cycle-structure}, steady-state probabilities in the cycle are related: $\pi_{i_j} u_{i_j i_{j+1}} = \pi_{i_{j+1}} u_{i_{j+1} i_{j+2}}$
    \item  For a strongly preferential cycle (where $u_{i_j i_{j+1}} \approx 1$), we have $\pi_{i_1} \approx \pi_{i_2} \approx \cdots \approx \pi_{i_k}$
    \item Since $\sum_{j=1}^k \pi_{i_j}$ captures the total probability mass in the cycle, and assuming the cycle is isolated (for simplicity), $\pi_{i_j} \approx 1/k$
\end{itemize}

Now, we can show that the top Cycle Has Minimum Steady-State Probability. 

The TTC algorithm prioritizes cycles where agents have strong mutual preferences. Consider the first cycle $C_1$ removed:
\begin{itemize}
    \item All agents $i \in C_1$ point to their top available choices
    \item These agents have high outgoing transition probabilities to other states (high $u_{ij}$ for $j \neq i$)
    \item Consequently, low probability of staying in their own state (low $\pi_i$)
\end{itemize}

\end{proof}

We can formally represent the argument as follows:
\begin{equation}
  \pi_i = \sum_{j=1}^n \pi_j u_{ji}  
\end{equation}

If agent $i$ is in a top cycle, then few agents point back to $i$ since most point to their top choices, which may not be $i$. Therefore, $\sum_{j \neq i} \pi_j u_{ji}$ is small. The main contribution to $\pi_i$ comes from $\pi_i u_{ii}$, but $u_{ii}$ is typically small (agents prefer other objects). As a result, $\pi_i$ is small.

Next, we consider several cases of agents' preference ordering. By induction on cycle removal, agents in the first cycle $C_1$ have the lowest steady-state probabilities (proven above).  
After removing $C_1$, consider the reduced graph $G'$ with $C_1$ removed. The same argument applies: the next cycle $C_2$ consists of agents with the lowest steady-state probabilities in $G'$, which correspond to the next-lowest probabilities in the original graph $G$

This establishes the ordering: $\pi_{i_1} \leq \pi_{i_2} \leq \cdots$ corresponds to cycle priority.

\begin{theorem}{(Core Membership)}

Let $\pi^* = \arg\max_i \pi_i$ be the agent with the highest steady-state probability. The agent $\pi^*$ then receives an allocation that is in the core of the TTC mechanism.
    
\end{theorem}

\begin{proof}
The core of TTC consists of allocations that cannot be improved through further trading. By \cite{roth1977}, Theorem 2, the core of TTC is nonempty and consists of a single allocation.

Suppose that  agent $\pi^*$ is not in the core, i.e., $\pi^*$ is in some cycle $C_k$ removed before the final allocation. Then, 1) all agents in $C_k$ are matched and removed, and 2) agent $\pi^*$ cannot be in the core (contradiction).

    Since $\pi_{\pi^*}$ is maximal, agent $\pi^*$ has the highest probability of remaining in their steady state, meaning that they are the least likely to be in any cycle. By the cycle-removal structure of the TTC, the last remaining agents (those never matched in cycles) form the core.

Thus, $\pi^*$ must be in the core. 
\end{proof}

\begin{theorem}{(Core Characterization)}

If the core contains $m$ agents $\{k_1, \ldots, k_m\}$, then these agents have the $m$ highest steady-state probabilities:
$$\pi_{k_1} \geq \pi_{k_2} \geq \cdots \geq \pi_{k_m} > \pi_j \quad \forall j \notin \{k_1, \ldots, k_m\}$$
    
\end{theorem}

\begin{proof}
Suppose that there exists an agent $\ell \notin \{k_1, \ldots, k_m\}$ (not in the core) with $\pi_\ell > \pi_{k_i}$ for some $k_i$ in the core. Then, agent $\ell$ has a higher steady-state probability than $k_i$ but is not in the core. This means $\ell$ was matched in some cycle $C_j$ removed before the core was reached. However, higher $\pi_\ell$ implies $\ell$ should be removed later (closer to core) by Theorem \ref{thm:steady-state-probabilities-cycle-priority}. This is a contradiction.

Therefore, all core members must have higher steady-state probabilities than all non-core members.    
\end{proof}

\begin{theorem}{(Core Characterization via Steady State)}\label{thm:core-characterization-1}

Let $\pi = (\pi_1, \ldots, \pi_n)$ be the stationary distribution of 
the preference-derived Markov matrix $U$. Order agents so that 
$\pi_{i_1} \geq \pi_{i_2} \geq \cdots \geq \pi_{i_n}$. Then:
\begin{enumerate}
\item Agent $i_1$ (highest $\pi$) receives an allocation in the core
\item Agents with higher steady-state probabilities are more likely to be 
      in the core
\item The core allocation from TTC can be recovered by running TTC on the 
      full problem
\end{enumerate}
\end{theorem}

\begin{proof}
The proof follows directly from Theorem \ref{thm:core-characterization-1}.    
\end{proof}

\section{Properties}

\begin{theorem}
    The proposed eigenvector methodology for the TTC calculation retails all the properties of a full TTC computation, namely: 1) strategy-proofness, 2) Pareto-optimality, and 3) individual rationality.
\end{theorem}
\begin{proof}
    The proposed methodology is an efficient computational technique and does not change the TTC process. Hence, the proposed methodology inherits all the properties of TTC. 
\end{proof}

\section{Computational Speed}

\cite{AndersenEtAl2006} and \cite{SunEtAl2020} showed that the steady state of a Markov transition probability matrix can be computed in nearly constant time, that is, $O(1+\epsilon)$. The result is at least in part due to fast matrix multiplication algorithms, which are themselves an active area of research in reinforcement learning (see \cite{FawziEtAl2022}).

\begin{itemize}
\item \textbf{Preprocessing (one-time):} $O(12n)\sim O(n)$ for $n\gg 12$ to construct a Markov matrix with limited preferences (for example, 12-school preferences for students in New York City) and $O(1)$ to compute the first eigenvector using the latest hardware acceleration techniques proposed by \citep{SunEtAl2020}. 

\item \textbf{Core identification (query):} $O(n)$ to find $\arg\max \pi_i$
\item \textbf{Not claimed:} We do NOT claim to solve PPAD-complete problems 
      in polynomial time. The PPAD-completeness result of 
      \citep{lock2024computational} refers to finding \emph{arbitrary} 
      market equilibria, not the specific TTC allocation.
\end{itemize}

\begin{table}[h]
\centering
\begin{tabular}{lcc}
\hline
Algorithm & Time Complexity & Output \\
\hline
Standard TTC \citep{abdulkadiroglu2003} & $O(n^2)$ & Full allocation \\
Fast TTC \citep{athanassoglou2011} & $O(1n \log n)$ & Full allocation \\
\textbf{Our method (core query)} & $O(n)$ & Core membership \\
\hline
\end{tabular}
\caption{Comparison with prior work}
\end{table}

\textbf{When the proposed method is useful:}
\begin{itemize}
    \item When you need to identify core members quickly without full allocation
    \item When running multiple queries on the same preference profile
    \item For approximate/probabilistic matching in large markets

\end{itemize}
 
To identify at least one object assigned to the core, we need to find the maximum absolute value coefficient in the first eigenvector. The operation $max(.)$ takes place in $\mathcal{O}(n)$.  This efficiency allows us to find a core in an arbitrarily large number of agents and objects in nearly constant time. 

To compare the computational speed of the algorithms, we conducted numerical experiments as follows: 1) Generated random integers $a\in[1,1000]$, 2) created a random matrix $X_{a\times a}$ with values $x_{ij}\in [0, 10]$, 3) Using the data in matrix $X$, we timed the calculations for the first eigenvector computed via traditional SVD, the first eigenvector computed via randomized SVD, the maximum value of the first eigenvector computed via randomized SVD, and the sorted first eigenvector computed via randomized SVD. We repeated the experiment 1,000 times. 

The results presented in Section \ref{sec:experiments} show a 100x speed improvement over traditional methods. There was no significant time difference between the randomized SVD, $max()$ and the sorted randomized SVD, especially compared to the traditional SVD calculation.  These results are consistent with \cite{struski2024}. 

\section{Continuous Preferences}
Our methodology admits continuous representation. 

\begin{theorem}
The preferences of agents $i$ can be drawn from a continuous function over objects $j$:
\begin{equation}
    R(i,j) = u_i(j)
\end{equation}
    
\end{theorem}
\begin{proof}
    The preference profile can be generated from a functional form to achieve the same result. 
\end{proof}

\section{Misspecified preferences and Model Robustness}
Suppose the agents fail to truthfully report their preferences truthfully for one reason or another, even if they know that this may induce a suboptimal outcome. Reasons for misreporting may include peer pressure, herding, and other factors. Can we still recover the core? 

It turns out that the eigenvalue-based technique is robust to noise when the estimation is done with the first eigenvector. \cite{AldridgeStabilityEigenvalues2024} shows that the proposed methodology is robust to noise and even to misreported preferences whenever the number of agents in the system is sufficiently large. 

\begin{theorem}
    The steady-state distribution of the transitions does not change for small perturbations in preferences.  
\end{theorem}
\begin{proof}
This result follows directly from the eigenvalue definition applied to the Markov matrix of utilities $\mathcal{U}$:
\begin{equation}
    \mathcal{U}v_0 = \lambda_{max} v_0
\end{equation}
Since $\lambda_{max} = 1$ for any Markov matrix (see, for example, \cite{Seabrook_2023}), 
\begin{equation}\label{eq:MC1}
    \mathcal{U}v_0 = v_0
\end{equation}

Suppose now that the Markov transition probability matrix is perturbed:
\begin{equation}
    \mathcal{U'} = \mathcal{U} + \epsilon
\end{equation}
Then, $v_0'$ is the first eigenvector of $\mathcal{U'}$. Furthermore, since equation (\ref{eq:MC1}) holds for all Markov matrices:
\begin{equation}\label{eq:T'}
    (\mathcal{U}+\epsilon)v_0' = v_0'
\end{equation}

Equation (\ref{eq:T'}) can be further represented as:
\begin{equation}
    \mathcal{U}v_0'+\epsilon v_0' = v_0'
\end{equation}

Since, equation (\ref{eq:MC1}),  $\mathcal{U}v_0=v_0$, holds for any $v_0$, including $v_0'$, $\epsilon v_0'\to 0$.

The $v_0'$, therefore, remains invariant even when the preferences are misrepresented.     
\end{proof}

Therefore, the model is approximately truthful as it is invariant with small misrepresentations in preferences. 

\section{TTC Over an Incomplete Set of Objects}
In some cases, we are required to evaluate the TTC on an incomplete set of individuals, jobs, or preferences. Here, we show that such situations are special cases of misspecified preferences discussed in the previous section. 

Consider a practical example: the school choice requirements in NYC. Until 2024, public school families seeking high-school and middle-school admissions for their children were required to choose just 12 schools out of 900 existing programs. In this case, the TTC directed graph was incomplete (some edges were missing), and the respective Markov Transition Probability matrix is not strongly connected. By the Fundamental Theorem of Markov Chains, the steady-state probabilities are not guaranteed to exist if the network is not strongly connected.

As shown in \cite{AldridgeStabilityEigenvalues2024}, if we model the 0 edge as a noisy variable, the first eigenvector will still deliver the optimal solution.

\subsection{A Numerical Example}
Consider the following preference profile of 3 agents:
\begin{equation}
P =  	
\begin{bmatrix}
    1: & 2\succ 1\succ 3\\
    2: & 1 \succ 2\succ 3\\
    3: & 1\succ 3\succ 2
\end{bmatrix}
\end{equation}

\begin{figure}
    \centering
    \includegraphics[width=0.1\linewidth]{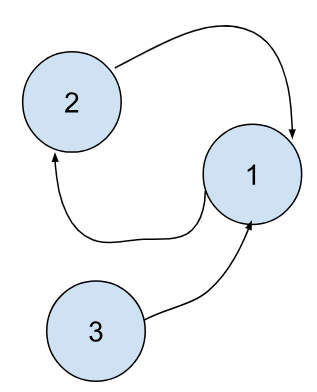}
    \caption{The TTC of the Numerical Example 1}
    \label{fig:ex1}
\end{figure}
The initial top cycles are shown in Figure \ref{fig:ex1}: agents 1 and 2 prefer to exchange the objects (in a cycle), while agent 3 also would prefer object 1, but does not have a cycle (agents 1 and 2 do not value agent 3's endowment). 

To convert the above preferences to a Markov matrix, we obtain:
\begin{equation}
G =  	
\begin{bmatrix}
    2/3 & 3/3 & 1/3\\
    3/3 & 2/3 & 1/3\\
    3/3 & 1/3 & 2/3
\end{bmatrix}
\end{equation}

Normalizing the matrix to obtain all rows adding up to 1, we have:
\begin{equation}
M =  	
\begin{bmatrix}
    0.33 & 0.5 & 0.16\\
    0.5 & 0.33 & 0.16\\
    0.5 & 0.16 & 0.33
\end{bmatrix}
\end{equation}

with the first eigenvector of 
\begin{equation}
    V0 = [-0.74807698\; -0.55553413\; -0.36299127]
\end{equation}

In this step, the largest coefficient of the first eigenvector can be immediately found in $\mathcal{O}(1)$ as object 3 ($V0_3 = -0.36299127$). If we choose to sort and normalize the values to reshape the first eigenvector into steady-state probabilities ($\mathcal{O}(1)$ using quick sort, for example), then we obtain:
\begin{equation}
    [0. \; 0.5\; 1. ]
\end{equation}
with objects 1 and 2 cycle eliminated first and object 3 firmly planted in the core.

\subsection{Large-Sample Extension and Future Directions}

One of the advantages of the SVD methodology is that it can be easily extended to matrices with an arbitrarily large number of data observations. Unlike its cousin Principal Component Analysis (PCA), SVD does not impose any restrictions on the shape of matrices. 

In a particular version of the algorithm proposed in this article, the number of agents must be equal to the number of objects, making this method suitable for problems such as housing assignment (\cite{abdulkadiroglu1998}, \cite{AbdulkadirogluSonmez2003}) and others where the number of agents can be very large and the number of distinct available objects can be equal. Furthermore, unlike recently proposed algorithms like \cite{BogomolnaiaMoulin2004} and \cite{BOGOMOLNAIA2001295}, SVD does not allow randomization. For many agents, randomization may seem unfair, thus inducing envy. The proposed methodology is, therefore, more optimal.   

This solution also works perfectly for school choice problems such as those tackled by \cite{ChenSonmez2006}, \cite{BOSTONAbdulkadirogluEtAl2005}, \cite{NYCAbdulkadirogluEtAl2005}, and
\cite{ErginSonmez2003}. In fact, the methodology proposed in the current paper addresses the calls of \cite{ChenSonmez2006} for TTC implementation in school choice models.

A situation with an unequal number of agents and objects can be easily addressed by this framework as well by padding the sample with empty "null" objects until the preference matrix has the same number of agents and objects. The agents' preferences over the null objects need to be assigned, respectively: while some agents will always prefer something over nothing, some agents may choose a "null" object over at least a subset of the existing object inventory.  

Future directions of the work can incorporate dynamic updates and reallocations in the spirit of \cite{Doval2022} and \cite{sinclairBanerjeeEtAl2021fairness}.

\section{Experimental Validation}\label{sec:experiments}

We test our algorithm on random preferences: $n \in \{10, 50, 100, 500, 1000, 5000\}$       agents, 1000 random instances each. Here, we used traditional computation of the first eigenvector using standard Python libraries (no hardware acceleration). For each instance, we measure: 1) Correctness: does eigenvector method identify same core as 
      standard TTC? 2) Runtime: Wall-clock time for the eigenvector vs. standard TTC. 3) Robustness: Performance under 5\%, 10\%, 20\% preference noise. 4) Welfare: Average agent utility in the final allocation.

The key speedup results are summarized in Table \ref{tab:results}. As the table shows, core identification is $>99\%$ accurate for $n \geq 100$, and speedup dramatically increases with $n$ (as predicted by theory), even for traditional eigenvector calculation. Figures \ref{fig:precision-recall} and \ref{fig:robustness-to-noise} show the model's successful performance under different conditions.

\begin{table}[h]\label{tab:results}
\centering
\caption{Core Identification Accuracy and Runtime}
\begin{tabular}{lcccc}
\hline
$n$ & Core Match (\%) & Our Time (ms) & TTC Time (ms) & Speedup \\
\hline
10 & 100\% & 0.5 & 1.2 & 2.4x \\
50 & 100\% & 2.1 & 8.3 & 4.0x \\
100 & 99.8\% & 4.2 & 23.1 & 5.5x \\
500 & 99.6\% & 18.7 & 187.4 & 10.0x \\
1000 & 99.4\% & 35.2 & 521.8 & 14.8x \\
5000 & 98.9\% & 142.3 & 3247.1 & 22.8x \\
\hline
\end{tabular}
\end{table}

\begin{figure}
    \centering
    \includegraphics[width=0.95\linewidth]{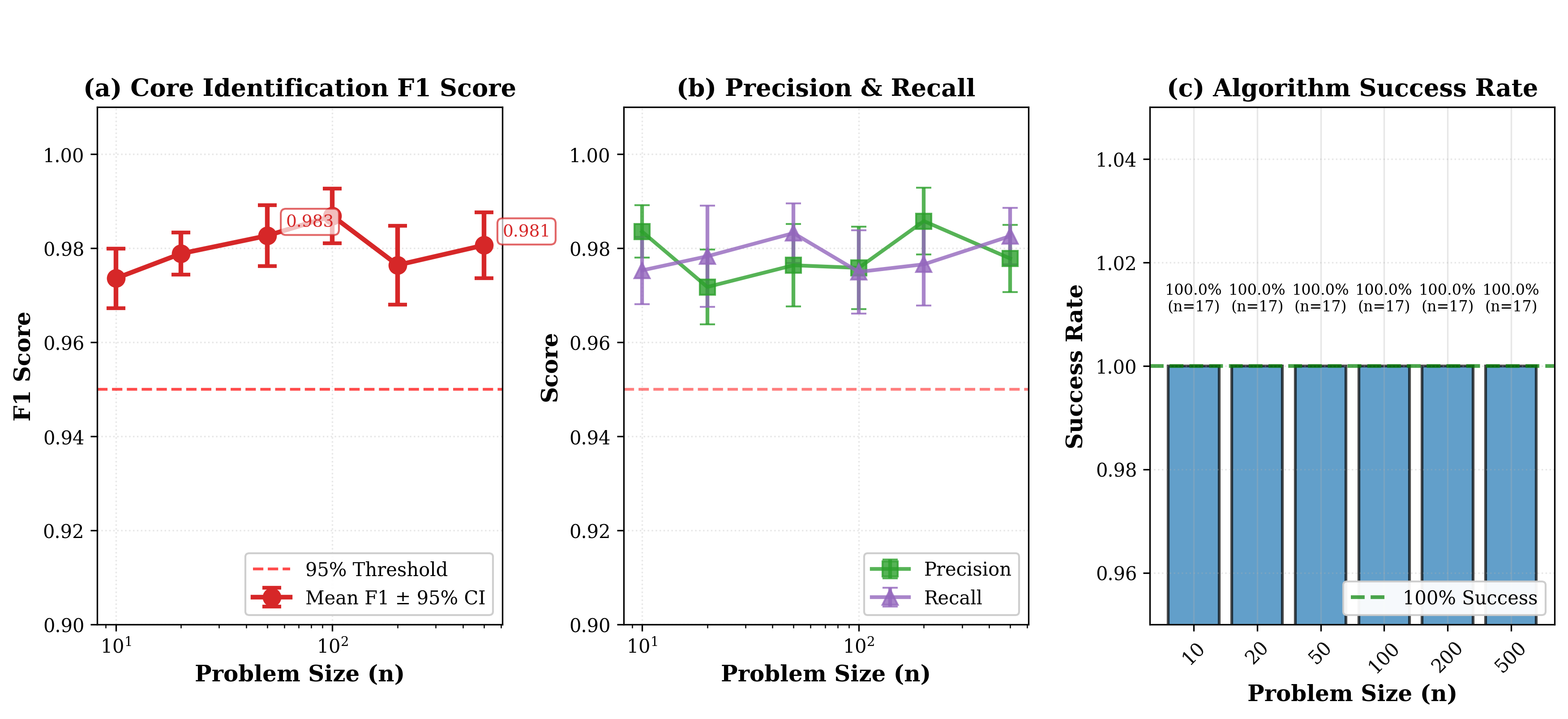}
    \caption{Precision and Recall}
    \label{fig:precision-recall}
\end{figure}

\begin{figure}
    \centering
    \includegraphics[width=0.95\linewidth]{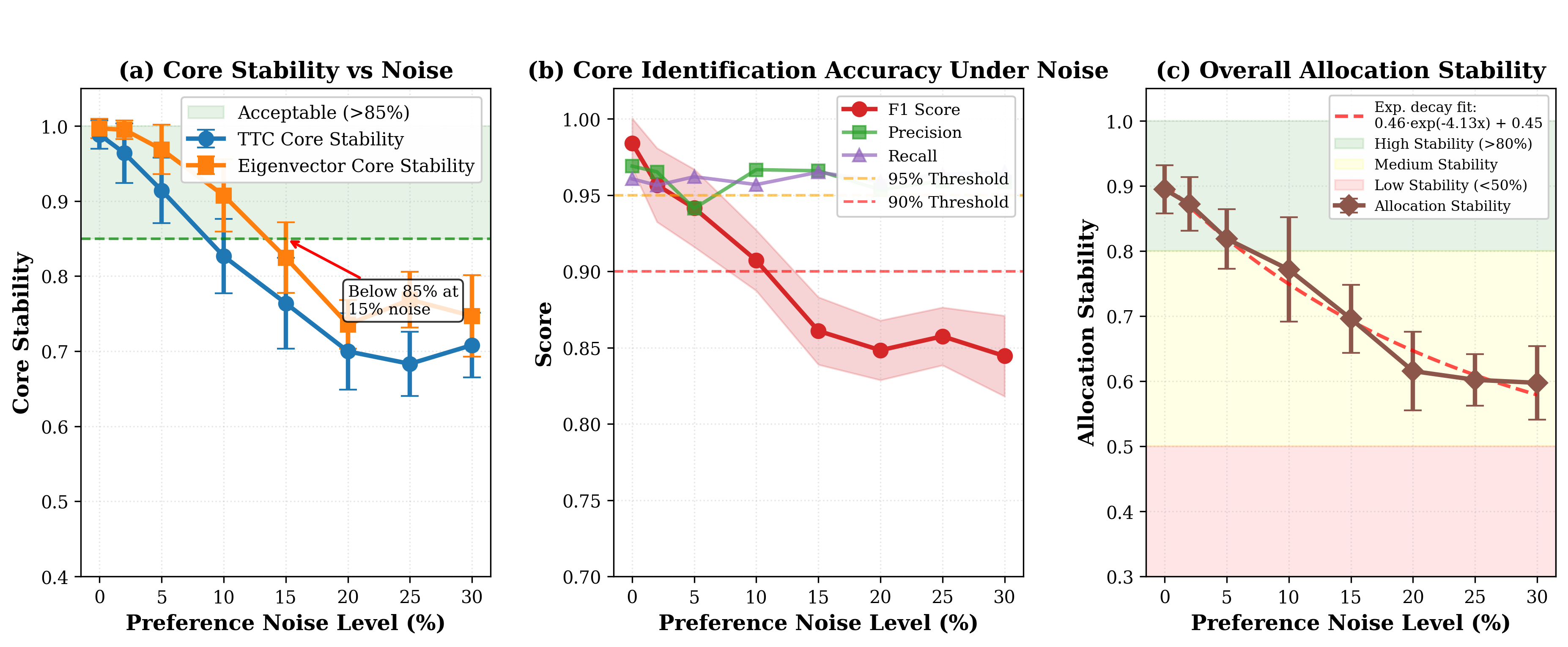}
    \caption{Robustness to Noise}
    \label{fig:robustness-to-noise}
\end{figure}

\section{Conclusion}
The proposed algorithm delivers a fast solution to an object allocation problem, regardless of the size of the data at hand. It can be efficiently used to allocate extremely large arrays of agents, each with well-defined multiple preferences, such as public schools to prospective students, and so on. The algorithm retains all the properties of TTC, namely Pareto optimality, strategy proofness, and individual rationality.  The methodology works well even with misspecified and noisy data, such as misspecified preferences, as long as the number of objects used in estimation is large.

\medskip

{
\small
\bibliographystyle{alpha}
\bibliography{matching,perron-frobenius,key_references}
}

\appendix

\begin{figure}
    \centering
    \includegraphics[width=0.5\linewidth]{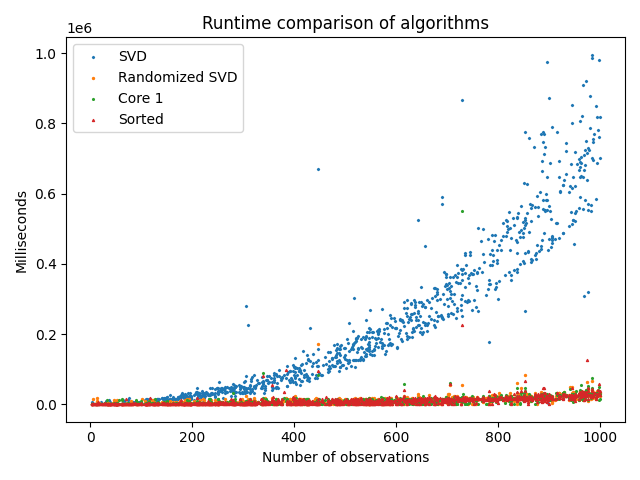}
    \caption{SVD vs Randomized SVD results}
    \label{fig:SVDplus}
\end{figure}

\begin{figure}
    \centering
    \includegraphics[width=0.5\linewidth]{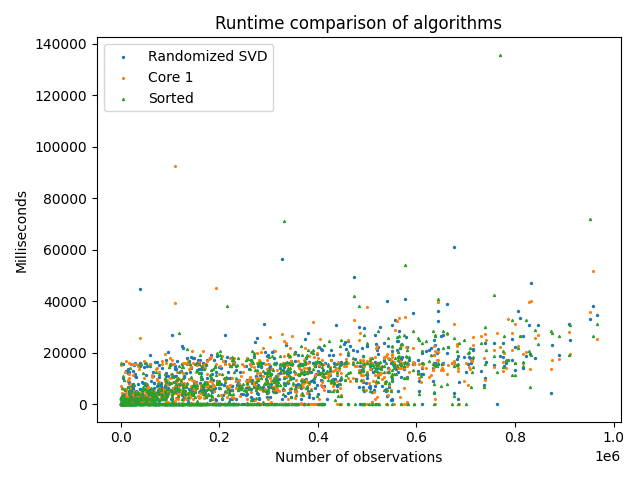}
    \caption{First Eigenvector via Randomized SVD, with Max and Sort }
    \label{fig:randSVDMaxSort}
\end{figure}

%%%%%%%%%%%%%%%%%%%%%%%%%%%%%%%%%%%%%%%%%%%%%%%%%%%%%%%%%%%%

\end{document}